\documentclass[12pt]{article}      
\usepackage{amsmath,amssymb,amsthm, amscd, graphicx, verbatim, setspace} 
\usepackage{setspace}

\theoremstyle{plain}
\newtheorem{thm}{Theorem}
\newtheorem{lem}[thm]{Lemma}

\newcommand\ceil[1]{\lceil#1\rceil}
\usepackage{algorithm}
\usepackage{algpseudocode}
\usepackage{amsthm}
\usepackage{amsmath}

\title{An anti-incursion algorithm for unknown probabilistic adversaries on connected graphs}
\date{}
\author{Jesse Geneson\\
\small\tt geneson@gmail.com
}
\begin{document}
\maketitle

\begin{abstract}
A gambler moves on the vertices $1, \ldots, n$ of a graph using the probability distribution $p_{1}, \ldots, p_{n}$. A cop pursues the gambler on the graph, only being able to move between adjacent vertices. What is the expected number of moves that the gambler can make until the cop catches them?

Komarov and Winkler proved an upper bound of approximately $1.97n$ for the expected capture time on any connected $n$-vertex graph when the cop does not know the gambler's distribution. We improve this upper bound to approximately $1.95n$ by modifying the cop's pursuit algorithm.
\end{abstract}

\section{Introduction}

Games with cops and robbers on graphs, which can be applied for designing anti-incursion programs, have been studied for several decades \cite{T1, T2, T3, T4, NW, Q}. We investigate a version of the game where the adversary moves among the vertices $1, \ldots, n$ following a probability distribution $p_{1}, \ldots, p_{n}$. Before the game starts, the cop picks and occupies a vertex from $G$. In each round of
the game, the cop selects and moves to an adjacent vertex or stays at the same vertex. The
gambler chooses to occupy a vertex randomly based on a time-independent distribution, not restricted to only adjacent vertices.

Whenever both players occupy the same vertex at the same time, the cop wins. The gambler is called a known gambler if the cop knows their probability distribution. Otherwise the gambler is called unknown. 

Gambler-pursuit games model anti-incursion programs navigating a linked list of ports, trying to minimize interception time for enemy packets. Komarov and Winkler proved that the expected capture time on any connected $n$-vertex graph is exactly $n$ for a known gambler \cite{KW}, assuming that both players use optimal strategies. For an unknown gambler, Komarov and Winkler proved an upper bound of approximately $1.97n$ \cite{KW}.

Komarov and Winkler conjectured that the general upper bound for the unknown gambler on a connected $n$-vertex graph can be improved from about $1.97n$ to $3n/2$, and that the star is the worst case for this bound. In Sections \ref{unkn} and \ref{unkn1}, we improve the upper bound for the unknown gambler's expected capture time to approximately $1.95n$ by using a different strategy for the cop.

\section{Unknown Gambler Pursuit Algorithm}\label{unkn}

Let $G$ be a connected $n$-vertex graph. As in \cite{KW}, let $T$ be a spanning subtree of $G$. We describe the cop's pursuit algorithm, and then we prove an upper bound of approximately $1.95n$ on the expected capture time.

Suppose that the cop performs a depth first search of $T$, except the cop stays at some leaves for two turns instead of one. Specifically, the cop uniformly at random selects a subset $U$ of $\ceil{0.72912 n}$ vertices and stays at the vertices in $U$ for an extra turn if the vertices are leaves. If there is a vertex $v$ in $U$ that is not a leaf, then the depth first search would already go twice through $v$, so the cop does not need to stay an extra turn at $v$. After the proof, we explain the reason for using the number $0.72912$.

The cop flips a coin to decide whether to perform the depth first search forward or backward. Thus the total number of turns in a single depth first search (including the extra turns for the leaves in $U$) is at most $1+2(n-1)+\ceil{0.72912 n} \leq 2.72912 n$. The search is repeated until capture. Since the cop flips a coin to decide whether to search forward or backward, the expected number of turns in the successful depth first search is at most $1.36456n$.

\section{Analysis}\label{unkn1}

Let the vertices of the graph be named $1, \ldots, n$. Suppose that the unknown gambler chooses vertex $i$ with probability $p_{i}$. We split the proof into two cases to show that the probability of evasion in a single depth first search is at most $0.17745$.

\begin{lem}
If there are two vertices $i$ and $j$ that the cop visits at least twice each such that $p_{i}+p_{j} \geq 0.732$, then the probability of evasion in a single depth first search is less than $0.162$.
\end{lem}

\begin{proof}
If the cop visits $i$ and $j$ both at least twice, then the probability of evasion is at most $(1-p_{i})^{2}(1-p_{j})^{2} \leq (1-p_{i})^{2}(0.268+p_{i})^{2}$, which has a maximum value of approximately $0.16157$ on the interval $[0,1]$ at $p_{i} = 0.366$.
\end{proof}

Next we show that the probability of evasion is at most $e^{-1.72912} < 0.17745$ when there are no vertices $i$ and $j$ that the cop visits at least twice each such that $p_{i}+p_{j} \geq 0.732$.

\begin{lem}\label{didj}
Suppose that there are no vertices $x$ and $y$ that the cop visits at least twice each such that $p_{x}+p_{y} \geq 0.732$. Then the probability of evasion for a single depth first search is at most $(1-\frac{1}{n})^{1.72912 n}$.
\end{lem}

\begin{proof}
Let $i, j$ be any two vertices of $G$ and suppose that $p_{i}+p_{j} = a$ and let $t_{1}, \ldots, t_{n-2}$ denote the vertices of $G$ not equal to $i$ or $j$. Given the condition that there are no vertices $x$ and $y$ that the cop visits at least twice each such that $p_{x}+p_{y} \geq 0.732$, then the probability of evasion for a single depth first search can be bounded by performing the following reduction to obtain shorter searches called $S$ and $S'$. 

First we define $S'$. If the cop visits a vertex $v$ not in $U$ more than once in the original depth first search, skip the cop's visits to $v$ after the first visit to $v$ in $S'$; if the cop visits a vertex $v$ in $U$ more than twice in the original depth first search, skip the cop's visits to $v$ after the second visit to $v$ in $S'$. Note that with $S'$, vertices in $U$ are visited exactly twice, and vertices not in $U$ are visited exactly once. The number of turns in $S'$ is thus $n+\ceil{0.72912n}=\ceil{1.72912n}$. To obtain $S$ from $S'$, skip all visits to vertices $i$ and $j$. 

Note that the reduction can only increase the probably of evasion, so the probability of evasion for the original depth first search is at most the probability of evasion for $S'$, which is at most the probability of evasion for $S$. Note also that the searches $S$ or $S'$ could be impossible for the cop to perform, since consecutive vertices in $S$ or $S'$ might not be adjacent. The searches $S$ and $S'$ are only used in this proof to obtain an upper bound on the probability of evasion for the original depth first search. 

For $c, d \in \left\{1,2\right\}$, define $f_{c,d}(p_{t_{1}}, \ldots, p_{t_{n-2}})$ to be the probability that the gambler evades the cop in search $S$ and that the cop makes $c$ visits to vertex $i$ and $d$ visits to vertex $j$ in search $S'$, conditioned on the fact that there are no vertices $x$ and $y$ that the cop visits at least twice each such that $p_{x}+p_{y} \geq 0.732$ in the original depth first search. Then the probability of evasion in search $S'$ is $p = p(p_{i},p_{j},p_{t_{1}},\ldots,p_{t_{n-2}})$ of the form $(1-p_{i})(1-a+p_{i}) f_{1,1}(p_{t_{1}}, \ldots, p_{t_{n-2}})+(1-p_{i})^{2}(1-a+p_{i}) f_{2,1}(p_{t_{1}}, \ldots, p_{t_{n-2}})+(1-p_{i})(1-a+p_{i})^{2} f_{1,2}(p_{t_{1}}, \ldots, p_{t_{n-2}})+(1-p_{i})^{2}(1-a+p_{i})^{2} f_{2,2}(p_{t_{1}}, \ldots, p_{t_{n-2}})$. Note that $f_{1,2} = f_{2,1}$ by symmetry and $p(\frac{1}{n},\ldots,\frac{1}{n}) = (1-\frac{1}{n})^{\ceil{1.72912 n}}$. 

Then $\frac{d p}{d p_{i}} = (a-2p_{i})f_{1,1}(p_{t_{1}}, \ldots, p_{t_{n-2}})+(a-2)(2p_{i}-a)f_{1,2}(p_{t_{1}}, \ldots, p_{t_{n-2}})+2(1-p_{i})(1-a+p_{i})(a-2p_{i})f_{2,2}(p_{t_{1}}, \ldots, p_{t_{n-2}})$, which is equal to zero at $p_{i} = \frac{a}{2}$. Moreover $\frac{d^{2} p}{d p_{i}^{2}} = -2 f_{1,1}(p_{t_{1}}, \ldots, p_{t_{n-2}})+2(a-2)f_{1,2}(p_{t_{1}}, \ldots, p_{t_{n-2}})+(12(p_{i}-\frac{a}{2})^{2}-(a-2)^{2})f_{2,2}(p_{t_{1}}, \ldots, p_{t_{n-2}})$, which is at most $-2 f_{1,1}(p_{t_{1}}, \ldots, p_{t_{n-2}})+2(a-2)f_{1,2}(p_{t_{1}}, \ldots, p_{t_{n-2}})+(2a^{2}+4a-4)f_{2,2}(p_{t_{1}}, \ldots, p_{t_{n-2}})$.

If $a \geq 0.732$, then $f_{2,2}(p_{t_{1}}, \ldots, p_{t_{n-2}}) = 0$ since there are no vertices $x$ and $y$ that the cop visits at least twice each such that $p_{x}+p_{y} \geq 0.732$. Thus $\frac{d^{2} p}{d p_{i}^{2}} = -2 f_{1,1}(p_{t_{1}}, \ldots, p_{t_{n-2}})+2(a-2)f_{1,2}(p_{t_{1}}, \ldots, p_{t_{n-2}}) \leq 0$ for all $p_{i} \in [0,a]$, so $p$ is maximized when $p_{i} = \frac{a}{2}$.

If $a \leq 0.732$, then $2a^{2}+4a-4 < 0$. Thus $\frac{d^{2} p}{d p_{i}^{2}} \leq 0$ for all $p_{i} \in [0,a]$, so $p$ is maximized when $p_{i} = \frac{a}{2}$.

This implies that $p$ is maximized only when $p_i=p_j$ for any two vertices $i$ and $j$ of $G$. Thus if there are no vertices $x$ and $y$ that the cop visits at least twice each such that $p_{x}+p_{y} \geq 0.732$, then $p$ is maximized when $p_{i} = \frac{1}{n}$ for all $i$, so the probability of evasion is at most $(1-\frac{1}{n})^{\ceil{1.72912 n}} \leq (1-\frac{1}{n})^{1.72912 n}$.
\end{proof}

By the last lemma, the probability of evasion for a single depth first search is at most $(1-\frac{1}{n})^{1.72912 n} \leq e^{-1.72912} < 0.17745$. Thus the expected number of depth first searches until the cop catches the robber is at most $\frac{1}{1-0.17745} < 1.21574$.

Let $X$ denote the random variable for the total number of turns in all of the cop's depth first searches, not including the successful depth first search. Let $Y$ denote the random variable for the number of turns in the successful search. By linearity of expectation, the expected capture time is equal to $E(X)+E(Y)$. 

We proved above that $E(X) \leq (2.72912 n)(1.21574-1) < 0.58879n$ and $E(Y) \leq 1.36456n$. Thus the upper bound on the expected capture time is less than $1.95335 n$.

\section{Comments}

The reason why we chose $0.72912$ for the constant in the pursuit algorithm was because the function $\frac{x}{1-e^{1-x}}-\frac{x}{2}$ has a minimum value of $1.95328$ on the interval $(1,\infty)$ at $x = 2.72912$. Despite this fact, it seems likely that our upper bound is not tight. 

The best current bounds on the maximum possible expected capture time for any connected $n$-vertex graph are between approximately $1.082n$ and $1.953n$. The lower bound follows from Komarov's proof for the cycle $C_{n}$ \cite{KT}.

However, there are a few families of graphs for which there are already tighter bounds. Komarov and Winkler proved an upper bound of approximately $1.082n$ for the expected capture time on the cycle $C_{n}$ \cite{KW} to match the lower bound from \cite{KT}. It is also easy to show that the expected capture time for the path $P_{n}$ is between $1.082n$ and $1.313n$ using the method in \cite{KW}. 

\section{Acknowledgments}
The author thanks Shen-Fu Tsai for helpful questions and suggestions on the exposition of this paper.

\end{document}